\documentclass{cccg25} 
\usepackage{graphicx,amssymb,amsmath}
\usepackage{cite,microtype,graphicx}

\usepackage[sectionbib]{bibunits}
\defaultbibliographystyle{abuser} 
\defaultbibliography{bottleneck}

\newif\iffull\fulltrue

\title{Decremental Greedy Polygons and Polyhedra Without Sharp Angles}

\author{David Eppstein\thanks{Department of Computer Science, University of California, Irvine. Research supported in part by NSF grant CCF-2212129.}}
\date{ }

\begin{document}
\maketitle  

\begin{abstract}
We show that the max-min-angle polygon in a planar point set can be found in time $O(n\log n)$ and a max-min-solid-angle convex polyhedron in a three-dimensional point set can be found in time $O(n^2)$. We also study the maxmin-angle polygonal curve in 3d, which we show to be $\mathsf{NP}$-hard to find if repetitions are forbidden but can be found in near-cubic time if repeated vertices or line segments are allowed, by reducing the problem to finding a bottleneck cycle in a graph.
We formalize a class of problems on which a decremental greedy algorithm can be guaranteed to find an optimal solution, generalizing our max-min-angle and bottleneck cycle algorithms, together with a known algorithm for graph degeneracy.
\end{abstract}

\begin{bibunit}

\section{Introduction}
In this work, we study  the natural problem of finding a polygon with vertices drawn from $n$ given points, maximizing its minimum (sharpest) angle (\cref{fig:maxmin-angle}). As we show, there exists an optimal polygon that is convex. To find it, we define the quality of a point, in a given subset, to be $2\pi$ if it is interior to the convex hull of the subset, or its interior angle if it is a vertex of the convex hull. This quality is monotonic: as we delete points, the quality of any remaining point can only decrease, as it becomes a hull vertex or as it loses hull neighbors. Therefore, it is safe to delete the point of minimum quality: any better polygon than the convex hull of the current subset cannot include the deleted point. A greedy algorithm that repeatedly deletes the sharpest hull vertex, and then returns the best polygon found throughout this deletion process, finds the maxmin-angle polygon in time $O(n\log n)$. After detailing this method we extend it to analogous problems of finding the max-min-solid angle convex polyhedron in 3d  We reduce max-min-angle polygons in 3d to finding bottleneck cycles in graphs, to which we apply related decremental greedy algorithms.

The decremental greedy nature of our algorithms both for geometry problems (max-min angle polygons and polyhedra) and graph problems (bottleneck cycles) suggests that they share a common generalization. We formalize a class of bottleneck optimization problems that includes these problems and can be solved optimally by decremental greedy algorithms. Our formalization also encompasses the known problem of computing graph degeneracy, the minimum degree of a vertex in a subgraph chosen to maximize this degree. A classical linear-time algorithm for degeneracy~\cite{MatBec-JACM-83} repeatedly removes a minimum-degree vertex until a given graph becomes empty; the degeneracy equals the maximum of the degrees of the vertices at the times of their removal.  Generalizing convex hull angles and vertex degrees to other measures of element quality, we define the \emph{bottleneck subset problem}, in which we seek a (nonempty) subset of a given set of elements whose worst element is as good as possible, according to a quality measure that can only worsen as other elements are removed. As we show, these problems can be solved by a decremental greedy algorithm that repeatedly removes the lowest-quality element.

\begin{figure}[t]
\centering\includegraphics[width=\columnwidth]{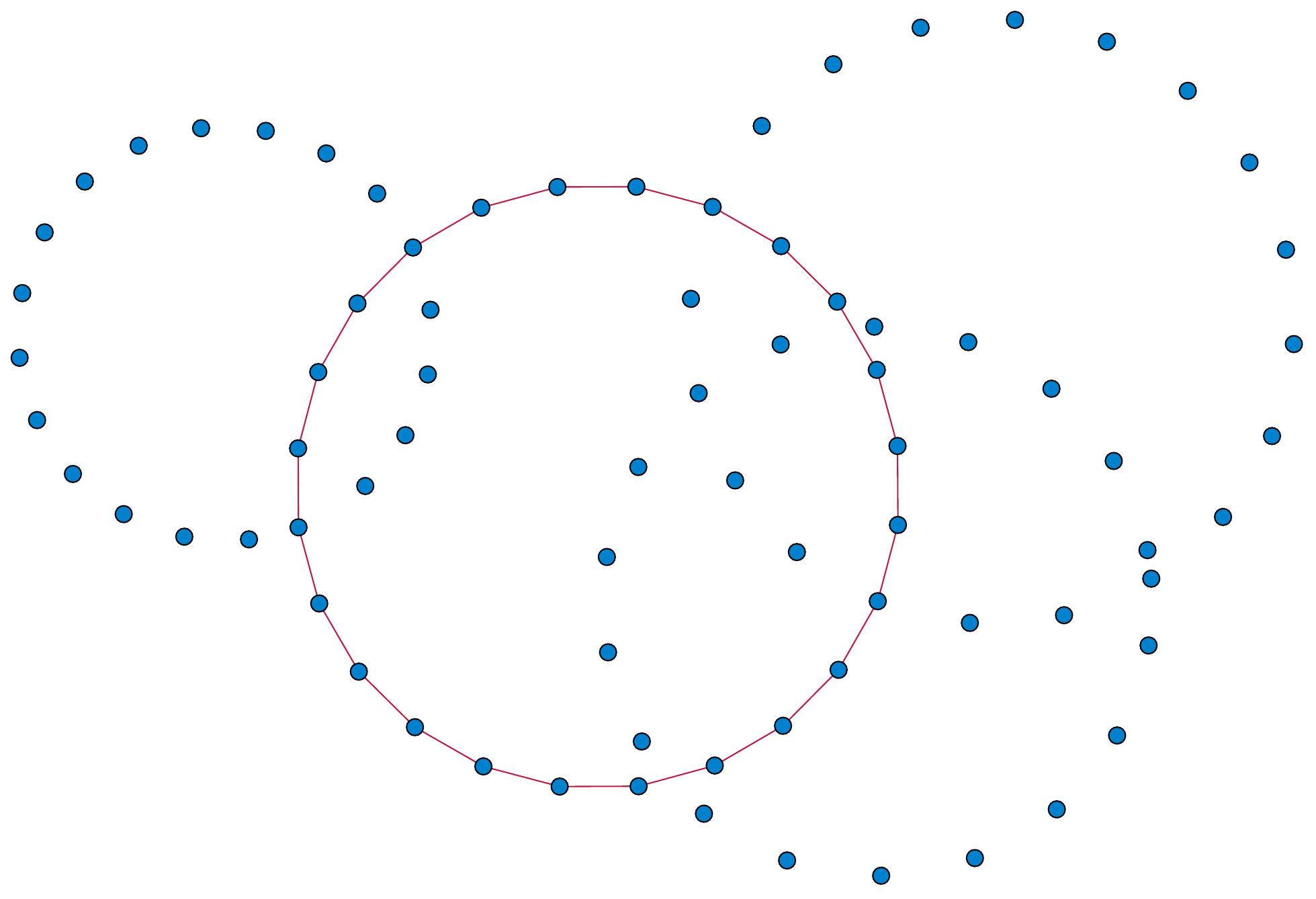}
\caption{The maxmin-angle simple polygon in a given set of points}
\label{fig:maxmin-angle}
\end{figure}

Although greedy algorithms are commonly associated with matroids, our formalization does not apply to decremental greedy algorithms for a max-min matroid base such as a maximum spanning tree, nor to greedy algorithms such as Dijkstra's algorithm. Intuitively, this difference comes from the direction of monotonicity. In decremental greedy matroid algorithms, elements become more valuable as other elements are removed and they become needed to complete a base, and in Dijkstra's algorithm, vertices in the priority queue become more valuable as better paths are found to reach them. In contrast, in the algorithms we study, elements become less valuable as other elements are removed. Our formalization is related to a different type of greedoid, an \emph{antimatroid}~\cite{KorLovSch-Greedoids-91}; see discussion following \cref{thm:antimatroid}.

\subsection{New results}
\cref{sec:formal} formalizes the bottleneck subset problem for monotonic quality measures, and describes the general decremental greedy algorithm for solving this problem. We provide in \cref{sec:geom} the following applications to geometric optimization:
\begin{itemize}
\item We prove that a maxmin-angle simple polygon in a planar point set~(\cref{fig:maxmin-angle}) can be chosen to be a convex polygon, and we show how to find it in time $O(n\log n)$.
\item We prove that a maxmin-solid angle polyhedral surface in a 3d point set can be chosen to be a convex polyhedron, and we show how to find it in time $O(n^2)$.
\item For a 3d point set, the maxmin-angle closed polygonal curve (by which we mean a cyclic sequence of line segments meeting end to end) may intersect at points or even entire line segments, so we do not call it a polygon. Even if it is a polygon, it may use points interior to its convex hull, or may be knotted. We show how to find a maxmin-angle closed polygonal curve in time $O(n^3\log^* n)$, or in time $O\bigl(n^3(\log n\log\log n)^2\bigr)$ if repeated points are allowed but repeated line segments are forbidden.
\end{itemize}

For space reasons we relegate additional results to 
\iffull
appendices. \cref{sec:antimatroid} details the antimatroid nature of the removals performed by the decremental greedy algorithm. In \cref{sec:np-hard} we outline a proof that the maxmin-angle non-self-intersecting curve in a 3d point set is $\mathsf{NP}$-hard to find. The same reduction also applies to the maxmin-angle non-self-intersecting unknotted curve.Our 3d polygonal curve results are based on decremental greedy algorithms for variations of bottleneck cycles in graphs, which we elaborate in \cref{sec:graphs}. 
\fi
For undirected graphs the bottleneck cycle problem has an easy non-decremental linear time algorithm. For other types of graphs, it becomes more complicated:
\begin{itemize}
\item In a directed graph with $m$ edges, or a mixed graph with $m$ directed and undirected edges, we show how to find a bottleneck cycle in time $O(m\log^* m)$.
\item We also consider \emph{polar graphs} or \emph{switch graphs}, in which each vertex has two poles at which edges attach, and a \emph{regular cycle} must pass through both poles of each of its vertices (\cref{fig:polarize}). We show how to find a bottleneck regular cycle in such a graph in time $O\bigl(m(\log m\log\log m)^2\bigr)$.
\end{itemize}

\subsection{Related work}
As far as we know the graph bottleneck cycle problems that we study are novel, but bottleneck path and bottleneck spanning tree problems were already studied by Pollack in 1960. The maximum spanning tree follows undirected bottleneck paths, and a variant of Dijkstra's algorithm constructs directed bottleneck paths~\cite{Pol-OR-60}.
Additional algorithms for these problems are known~\cite{GabTar-Algs-88,DuaPet-SODA-09}, and they have several applications~\cite{Sch-SCW-11,UllLeeHas-ICCAD-09,FerGarArb-OR-98}.

Connecting given points into curves or surfaces has been studied with the goal of reconstructing an unknown shape from sparse samples~\cite{Att-CGTA-98,AmeBerEpp-GMIP-98,OhrPeePar-CGF-21}, in some cases assuming that the curve is sampled densely enough to cause all angles to be close to $\pi$~\cite{AmeBerEpp-GMIP-98}. The problems that we study have a similar flavor, but for curves through a subset of points rather than requiring a curve to pass through or near all points. We are unaware of previous work using the max-min-angle optimization criterion for curves and surfaces, but this criterion is well known in computational geometry in the context of Delaunay triangulations, which maximize the minimum angle among triangulations of given planar points~\cite{Sib-CJ-78}, and has also been applied to other forms of triangulation~\cite{Mit-CGTA-97,Joe-IJNME-91,AmeBerEpp-Algs-99}. Maximizing the minimum angle is also important  in graph drawing, where the minimum angle at any vertex of a graph drawing is its \emph{angular resolution}~\cite{ForHagHar-SICOMP-93,MalPap-SIDMA-94}.

Dynamic programming can find convex polygons with vertices in a planar point set, optimizing a broad range of criteria~\cite{ChvKli-CN-80,AviRap-SoCG-85,EppOveRot-DCG-92}. However, we do not require the maxmin-angle polygon to be convex (instead this is an emergent property of the result) and our algorithm is simpler and faster than the known dynamic programming methods. It is not obvious how to generalize the dynamic programs for optimal convex polygons to curves or surfaces in higher dimensions.

\section{Formalization}
\label{sec:formal}

We define a \emph{monotone bottleneck subset problem} to consist of a set $U$ of elements, and a function $q(x,S)$  that takes as input a pair $(x,S)$ where $x\in S$ and $S\subset U$, and produces as output a real number, the quality of~$x$ as a member of $S$. We require $q$ to be monotone: whenever $x\in S\subset T\subset U$, we have $q(x,S)\le q(x,T)$. Intuitively, removing other elements from a subset containing $x$ causes the quality of $x$ to decrease or stay the same.  We define the quality $Q(S)$ of a nonempty subset $S$ to be $Q(S)=\min_{x\in S} q(x,S)$, the least quality of a member of~$S$, with $Q(\emptyset)=-\infty$, a flag value preventing the empty set from being optimal. Our goal is to find a nonempty subset whose quality is maximum, which we call the \emph{bottleneck subset}. If $x$ is a lowest-quality element of a subset $S$ we call $x$ a \emph{bottleneck element} of $S$.

The \emph{decremental greedy algorithm} for a monotone bottleneck subset problem performs the following steps:

\newpage
\begin{enumerate}
\item Initialize two sets $S$ and $T$  to both equal $U$.
\item While $S$ is not empty, repeat the following steps:
\begin{itemize}
\item If $Q(S)>Q(T)$, set $T=S$.
\item Find any $x$ with $q(x,S)\le Q(T)$.
\item Remove $x$ from $S$.
\end{itemize}
\item Return $T$.
\end{enumerate}

Any bottleneck element $x$ of~$S$ satisfies $q(x,S)\le Q(T)$, so the algorithm always has an element that it can remove, although we do not require it to remove a bottleneck element at each step. When the bottleneck quality $\beta$ is already known, a simpler decremental algorithm can find the bottleneck subset: repeatedly remove any element of quality less than $\beta$ until all remaining elements have greater value. Call this the \emph{known-$\beta$ algorithm}.

\begin{theorem}
\label{thm:antimatroid}
Every monotone bottleneck subset problem has a unique maximal bottleneck subset. Regardless of the choices they make at each step, the decremental greedy algorithm and known-$\beta$ algorithm both always find and return the maximal bottleneck subset.
\end{theorem}

\begin{proof}
The union $M$ of all bottleneck subsets is a bottleneck subset: each $x\in M$ belongs to a bottleneck subset $X\subset M$, so $q(x,M)\ge q(x,X)\ge\beta$. Because this is true for all elements of $M$, $Q(M)\ge\beta$, the optimal value. $M$~is the unique maximal bottleneck subset, because it is a superset of all other bottleneck subsets.

Because each proper superset $V\supset M$ is not a bottleneck subset, some element $y\in V$ has $q(y,V)<\beta$, and any such $y$ cannot be in $M$ by monotonicity of its quality. Therefore, until $M$ is reached, the known-$\beta$ algorithm can and will remove an element of the complement of $M$. And until $M$ is reached, the decremental greedy algorithm will have $Q(T)<\beta$ and the element $x$ that it removes will have $q(x,S)<\beta$; again, $x\notin M$ by monotonicity of its quality. Therefore, until $M$ is reached, the decremental greedy algorithm can and will remove an element of the complement of $M$.

Once $M$ is reached, the known-$\beta$ algorithm terminates and returns it. The decremental greedy algorithm records $M$ as the set $T$ that it will eventually return; it can never subsequently change $T$ to another set because no other set has better quality.
\end{proof}

Both algorithms make arbitrary choices that can cause them to produce different removal sequences before reaching $U$. Their families of allowed removal sequences form \emph{antimatroids}, structures that formalize the familiar introductory programming concept of a ready list. Antimatroids can be defined as families of sequences generated by a process that repeatedly appends an arbitrary ``available'' element from a given set, under the constraint that availability is determined by a monotonic function of the elements that have already been appended. Less formally, once an element becomes available, it remains available until it is appended itself, and availability depends only on the set of elements that have been chosen, not on their order. For the known-$\beta$ algorithm, the antimatroid property is straightforward (the property of having quality at most $\beta$ is determined by a monotonic function, as required) but for the decremental greedy algorithm, it requires a proof; see
\iffull
\cref{sec:antimatroid}.
\fi

According to the monotonicity that we require our quality measures to satisfy, each element gets worse (prioritized for earlier removal) as other elements around it are removed. This behavior should be contrasted with the quality of elements in max-min matroid problems such as the maximum spanning tree, which can also be solved by a decremental algorithm (remove the minimum-weight non-bridge edge until all remaining edges are bridges). In the decremental maximum spanning tree algorithm, an element (an edge) either keeps its priority (its weight) or gets a better priority (it becomes an unremovable bridge) as the algorithm progresses, rather than getting worse, so the decremental greedy maximum spanning tree algorithm does not fit into our framework.

Finding graph degeneracy is a bottleneck subset problem where $U$ is the set of vertices of a given graph $G$ and $q(x,S)=\deg_{G[S]}x$ is the degree of vertex $x$ in the subgraph induced by $S$. The bottleneck elements of any induced subgraph are its minimum-degree vertices. The bottleneck subset is a set of vertices that induces a subgraph maximizing its minimum degree.
The linear-time degeneracy algorithm of Matula and Beck~\cite{MatBec-JACM-83} repeatedly removes a minimum-degree vertex, as a special case of the decremental greedy algorithm for this quality. However, even for graph degeneracy, the antimatroid nature of the decremental greedy algorithm appears to be novel. To achieve this antimatroid property, we require a more general algorithm, allowing the removal of any vertex whose degree is at most the current quality bound, rather than the special case that only removes minimum-degree vertices.

\section{Geometric applications}
\label{sec:geom}

\subsection{Polygons in 2d}
\label{sec:2d}

\begin{lemma}
\label{lem:walk-is-convex}
Let $S$ be a finite set of points in the plane. Then there exists a convex polygon $P$ with vertices in~$S$ that maximizes the minimum angle among all closed polygonal curves (allowing repeated vertices and edges) with vertices in $S$.
\end{lemma}

\begin{proof}
Because $S$ determines finitely many angles, the max-min angle among closed polygonal curves exists. Let $W$ be any closed polygonal curve through $S$ attaining this angle, and let $P$ be the convex hull of $W$. Then, compared to $W$, $P$ may omit some vertices and may increase the angle of the others that remain; both of these changes can only increase the sharpest angle in $P$, relative to the sharpest angle in $W$. Therefore, $P$ is also a max-min angle closed polygonal curves, as was stated to exist in the lemma.
\end{proof}

\begin{theorem}
Let $S$ be a finite set of points in the plane. Then in time $O(n\log n)$ we can find a convex polygon $P$ with vertices in $S$ that maximizes the minimum angle among all closed polygonal curves in $S$.
\end{theorem}

\begin{proof}
We perform the following steps:
\begin{enumerate}
\item Initialize a dynamic convex hull data structure to contain all points of $S$.
\item Initialize parameters $\theta$ to the sharpest angle of the hull, and $s$ to zero. These parameters will store the best angle found so far, and the number of removed points corresponding to that best angle.
\item Initialize an empty list $L$ of removed points.
\item While the current convex hull is non-degenerate (it has more than two vertices), repeat the following steps:
\begin{itemize}
\item Find and remove from the dynamic hull the vertex with the sharpest angle (choosing arbitrarily any vertex of sharpest angle in case of ties) and append this vertex to $L$.
\item Set $\theta$ to the maximum of its previous value and the sharpest angle of the current hull, and if the result is an increase in $\theta$ then set $s$ to the current length of $L$.
\end{itemize}
\item Return the convex hull of the point set obtained from $S$ by removing the first $s$ points of $L$.
\end{enumerate}

This is a decremental greedy algorithm where the quality of a point is its interior angle, if it is a convex hull vertex, or $2\pi$ otherwise. This quality is monotonic and the algorithm finds a max-min angle convex polygon by \cref{thm:antimatroid}. This polygon is also a max-min angle closed polygonal curve  by \cref{lem:walk-is-convex}.

The sharpest angle in the current hull can be maintained during this algorithm using a priority queue of the current convex hull vertices and their angles, updated whenever the dynamic convex hull structure adds a vertex to the hull or changes the neighbor of an existing vertex. This structure requires $O(n)$ updates over the course of the algorithm and therefore takes $O(n\log n)$ time.
Decremental or fully dynamic convex hull data structures that take $O(\log n)$ time per update are also known~\cite{HerSur-BIT-92,BroJac-FOCS-02}, leading to an $O(n\log n)$ time bound for that part of the algorithm as well.
\end{proof}

\subsection{Polyhedra in 3d}

In this section we seek a polyhedral surface, with vertices at a subset of a given point sets, maximizing the solid angle interior to the surface as viewed from any point (or vertex) of the surface. To avoid definitional issues we consider only non-self-intersecting surfaces.
Analogously to the results of the previous section, we have:

\begin{lemma}
Among non-self-intersecting polyhedral surfaces having vertices at a subset of given points in~$\mathbb{R}^3$,
there exists a convex polyhedron that maximizes the minimum solid angle.
\end{lemma}

\begin{proof}
As in \cref{lem:walk-is-convex}, consider any polyhedral surface that maximizes the minimum solid angle, and take its convex hull. This can only remove vertices and improve the solid angle at the remaining vertices, so it must also maximize the minimum solid angle.
\end{proof}

As in the two-dimensional case, we will need a data structure for decremental convex hulls. Chan has studied this problem~\cite{Cha-JACM-10,ChaTsa-DCG-16}, but his algorithms do not represent the hull explicitly, instead using an implicit representation that allows only extreme point queries. In our case, we need to find the solid angles of the vertices of the hull, not possible with Chan's structure. Another data structure, of Buchin and Mulzer~\cite{BucMul-JACM-11} allows hulls of any subset of an input point set to be computed in randomized expected time $O\bigl(n(\log\log n)^2\bigr)$ on a word RAM.
Instead, we use a simpler method for maintaining three-dimensional convex hulls explicitly, in time $O(n)$ per point deletion. Such a data structure was described by Overmars~\cite{Ove-83} (Theorem 6.4.1.6, p. 90), but to keep the presentation self-contained we outline a simplified version. The simplified data structure consists of a balanced binary search tree for the $x$-coordinate order of the given points, together with  explicitly represented hulls of the sets of not-yet-deleted points in each subtree. After each deletion, each changed hull can be computed by merging two child hulls in time linear in its subtree. The total time for all merges adds in a geometric series to $O(n)$ per deletion. The space for this data structure is $O(n\log n)$. The version of Overmars improves the space to $O(n\log\log n)$ by storing only the hulls of large subsets of points, and makes the data structure fully dynamic rather than decremental using weight-balanced trees rather than static balanced trees, but we do not need those advances.

\begin{theorem}
Let $S$ be a finite set of points in $\mathbb{R}^3$. Then in time $O(n^2)$ we can find a convex polyhedron $P$ with vertices in $S$ that maximizes the minimum solid angle among all non-self-intersecting polyhedral surfaces having vertices at a subset of $S$.
\end{theorem}

\begin{proof}
We perform the following steps:
\begin{enumerate}
\item Initialize a dynamic convex hull data structure for the points of $S$.
\item Initialize parameters $\theta$ to the sharpest solid angle of the hull, and $s$ to zero. These parameters will store the best angle found so far, and the number of removed points corresponding to that best angle.
\item Initialize a list $L$ of removed points to an empty list.
\item While the current convex hull is non-degenerate (it has nonzero volume), repeat the following steps:
\begin{itemize}
\item Find and remove the vertex with the sharpest solid angle (choosing arbitrarily any vertex of sharpest angle in case of ties) and append this vertex to $L$.
\item Set $\theta$ to the maximum of its previous value and the sharpest angle of the current hull, and if the result is an increase in $\theta$ then set $s$ to the current length of $L$.
\item Update the hull of the remaining points.
\end{itemize}
\item Return the convex hull of the point set obtained from $S$ by removing the first $s$ points of $L$.
\end{enumerate}

Using the dynamic hull data structure outlined above, each iteration of the main loop takes time $O(n)$, and the whole algorithm takes time $O(n^2)$.
\end{proof}

\subsection{Closed polygonal curves in 3d}

To model the search for a closed polygonal curves through a given system of line segments in 3d, allowing repeated vertices but without allowing any line segment to be repeated, we use \emph{polar graphs}. These are graphs in which the edges are undirected, but are attached to one of two \emph{poles} of each vertex. A \emph{regular cycle} in such a graph is a simple cycle for which the two edges incident to each vertex are attached to different poles: if the cycle enters a vertex via one pole, it must exit via the other pole. In
\iffull
\cref{sec:polar}
\fi
we describe a decremental greedy algorithm for finding a minmax-weight regular cycle in a polar graph with $m$ edges in time $O(m(\log m\log\log m)^2)$. For a given set $S$ of points in $\mathbb{R}^3$, we may define a weighted polar graph $G(S)$, as follows (\cref{fig:polarize}):
\begin{itemize}
\item The vertices of $G(S)$ are the line segments determined by pairs of points in $S$.
\item The two poles of each vertex of $G(S)$ are the two endpoints of the corresponding line segment.
\item The edges of $G(S)$ connect pairs of line segments that form a three-point polygonal chain and are weighted by the angle at the middle point of the chain.
\item For each vertex and incident edge in $G(S)$, the pole of the vertex to which the edge is attached is the middle point of the three-point polygonal chain that defines the edge.
\end{itemize}

\begin{figure}[t]
\centering\includegraphics[scale=0.25]{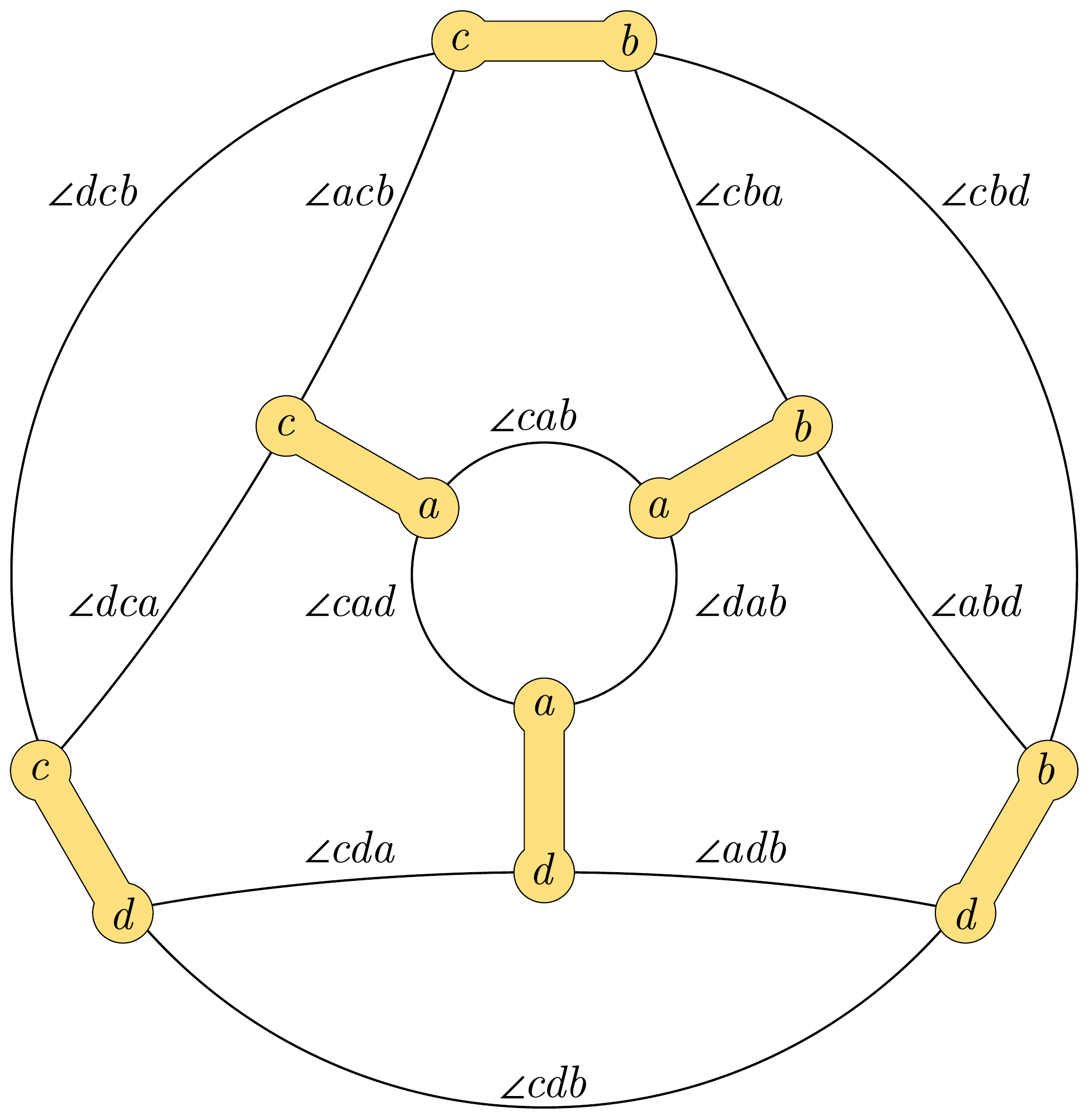}
\caption{A polar graph with six vertices representing the line segments determined by pairs of four points $a$, $b$, $c$, and $d$, and 12 edges labeled by the angles determined by triples of points. The two poles of each vertex are labeled by the points. A regular cycle must enter and exit the vertices it visits via opposite poles.}
\label{fig:polarize}
\end{figure}

Then, a regular cycle in this graph corresponds to a closed polygonal curve in 3d, allowing repeated points but not allowing repeated line segments. The minimum edge weight in this cycle is the sharpest angle of the curve. If repeated line segments are to be allowed, we can instead replace the polar graph by an ordinary graph, its \emph{double cover}. This has two copies of each vertex of the polar graph,
and two directed edges for each edge of the polar graph, in each direction. Each copy of a vertex is incident to the incoming directed edges for one pole of the vertex and to the outgoing directed edges for the other pole. The resulting directed graph is \emph{skew-symmetric} (the vertex bijection that swaps the two copies of each polar vertex also reverses all the directed edges) and has been used in algorithms for polar
\iffull
graphs; see \cref{sec:polar}.
\fi
When a directed cycle in the double cover of $G(S)$ uses both copies of a vertex, it corresponds to a closed polygonal curve that uses a line segment twice, in both directions. In
\iffull
\cref{sec:directed}
\fi
we describe a decremental greedy algorithm for finding a bottleneck directed cycle in a directed graph in time $O(m\log^* m)$.

\begin{theorem}
We consider maxmin-angle closed polygonal curves in a 3d point set that may use points interior to its convex hull, may be knotted, and may self-intersect. If both repeated points and repeated line segments are allowed, we can find such a curve in time $O(n^3\log^* n)$. If repeated points but not repeated line segments are allowed, we can find it in time $O\bigl(n^3(\log n\log\log n)^2\bigr)$.
\end{theorem}

\begin{proof}
We construct $G(S)$, with $O(n^3)$ edges, and then find either a bottleneck directed cycle in the double cover of $G(S)$, or a bottleneck regular cycle in $G(S)$, according to whether we allow or disallow repeated segments, respectively.
\end{proof}

For completeness, we state:

\begin{theorem}
\label{thm:hard}
It is $\mathsf{NP}$-complete to find the maxmin-angle simple polygon for points in~$\mathbb{R}^3$, for points on a unit sphere given by rational Euler angles.
\end{theorem}

We defer the proof to
\iffull
\cref{sec:np-hard}.
\fi

\section{Conclusions}

We have formalized a broad family of decremental greedy problems for monotone bottleneck subset problems, generalizing existing graph degeneracy algorithms, and shown its applicability in developing new algorithms in computational geometry and graphs. We have found simple greedy decremental algorithms for several maxmin-angle problems in geometric optimization: finding a planar polygon through given points maximizing the minimum angle, finding a 3d polyhedral surface through given pints maximizing the minimum solid angle, and finding a (self-intersecting) 3d closed polygonal curve through given points maximizing the minimum angle. On the other hand, constraining the 3d polygonal curve to be non-self-intersecting appears to make the problem computationally infeasible.

There appears to be room for improvement in our polyhedral surface and 3d curve time bounds. It would also be of interest to find other problems to which the same greedy decremental approach applies.

\putbib
\end{bibunit}

\iffull

\newpage
\appendix

\begin{bibunit}

\section{Antimatroid property of the decremental greedy algorithm}
\label{sec:antimatroid}

It is not completely obvious that the removal sequences of the decremental greedy algorithm form an antimatroid. The algorithm clearly meets the requirement that an element, once ready to be removed, remains ready until it is removed, but there is a second requirement: the availability of an element to be removed must depend only on the set of prior removals, and not on their ordering. For the decremental greedy algorithm,  the state of the variable $T$ appears to depend on the sequence of removals and not merely on the set of removals. Nevertheless, we prove that this algorithm defines an antimatroid, by using a structure of nested sets generalizing the $k$-cores of an undirected graph. (The $k$-cores are maximal subgraphs of minimum degree $k$, obtained by a generalization of the Matula--Beck degeneracy algorithm.)

\begin{lemma}
For any instance of the bottleneck subset problem, there exists a function $t$ such that, no matter which removal sequence is chosen by the decremental greedy algorithm, whenever the algorithm chooses the element $x$ to remove we have the equality $T=t(S)$. This function is monotonic in the sense that $S\subseteq S'\Rightarrow t(S)\subseteq t(S')$ and $S\subseteq S'\Rightarrow Q(t(S)\ge Q(t(S'))$.
\end{lemma}

\begin{proof}
For any set $X$ with quality $Q(X)$ less than the optimal quality of the given instance, let
\[\operatorname{better}(Q)=\min\{Q(Y)\mid Q(Y)>Q(X) \land Y\subset X\},\]
the smallest improved quality that can be obtained from a subset of $S$, and let $X^+$
denote the set
\[X^+=\bigcup\{ Y\mid Q(Y)\ge \operatorname{better}(Q),\]
the union of subsets of $X$ with this quality. Note that $Q(X^+)=\operatorname{better}(Q)$.

Then, whenever the decremental greedy algorithm finds a set $S$ that improves the current value of $Q(T)$, and sets $S=T$,
it must eventually reach the set $S^+$. For, until it does, the only elements that it can remove are elements with quality at most $Q(T)$,
and because this quality is less than $Q(S^+)$ the removed elements cannot belong to $S^+$.
It follows that, starting from $U$, the removal sequence must pass through the family of sets
\[ U, U^+, U^{++}, \dots, \] and that the current set $T$ must be the one of the sets in this sequence.

Then we can define $t(S)$ to be the last set in this sequence for which $S$ is a subset.
\end{proof}

\begin{theorem}
For any monotone bottleneck subset problem, the allowed removal sequences of the decremental greedy algorithm form an antimatroid.
\end{theorem}

\begin{proof}
For each element $x$, $x$ is removable under the monotonic condition that $q(x,S)\le Q\bigl(t(S)\bigr)$. Because we have formulated removability for this algorithm as a monotonic Boolean combination of the elements that have been removed so far,
it follows that the removal sequences form an antimatroid.
\end{proof}

Every antimatroid can be represented in this way by a monotone bottleneck subset problem in which $q(x,S)$ is 0 if element $x$ is available to be the next element in a sequence, when the remaining elements are $S$, and 1 if it is not available. For this quality function, the allowed removal sequences of both the decremental greedy algorithm and the fixed-$\beta$ algorithm with $\beta=0$ are exactly the allowed sequences of the given antimatroid.

\section{Hardness for 3d simple polygons}
\label{sec:np-hard}

In this section, we provide a proof that It is $\mathsf{NP}$-complete to find the maxmin-angle simple polygon for points in~$\mathbb{R}^3$, for points on a unit sphere specified by rational Euler angles (\cref{thm:hard}). This input assumption allows us to place points within coordinate systems that, anywhere away from the poles of the sphere, approximate a square grid, simplifying the analysis of the numerical precision needed for the construction.
We believe that it should be possible to instead assume that the points lie on a unit sphere and are specified by rational Cartesian coordinates, but we have not worked out the details needed to use this alternative representation.

As usual for $\mathsf{NP}$-hardness proofs, we use a polynomial-time many-one reduction from a known $\mathsf{NP}$-complete problem, 3-satisfiability.  We will translate the variables and clauses of a given 3-satisfiability instance into an instance of the maxmin-angle simple polygon problem and an angle threshold $\theta$ such that the optimal angle is $\ge\theta$ for yes-instances of the 3-satisfiability problem and $<\theta$ for no-instances of the 3-satisfiability problem. More strongly, we will perform the reduction in such a way that there is a gap between the angles obtained from a satisfying assignment and the angles obtained from an invalid solution. This gap will allow us to round the points of our reduced instance to rational coordinates without affecting its correctness.

To construct our maxmin-angle simple polygon problem instance, we translate the variables and clauses of a given 3-satisfiability instance into gadgets connected by smooth curves on a unit sphere, and then sample those curves by points spaced at distance at most $\varepsilon$ (for some small-enough value of the parameter $\varepsilon$. This will ensure that, if the curves can be connected into a single non-self-intersecting curve through the gadgets), it will form a polygonal curve with all angles at least some threshold $\theta=2\pi-\Theta(\varepsilon)$. However, we will arrange our reduction so that unwanted segments between pairs of points in unrelated gadgets will be separated at a distance $\Omega(\epsilon)$ with a large enough constant that they form an angle sharper than $\theta-\Omega(\varepsilon)$ with respect to the sphere (where the $\Omega(\varepsilon)$ term is the gap between the angles of valid and invalid solutions). In this way, any polygonal curve using these unwanted segments will also have an angle sharper than $\theta$, allowing us not to worry about any unwanted non-local interactions between parts of our construction.

The overall form of the reduction is one used for instance for a different $\mathsf{NP}$-completeness proof by Lynch \cite{Lyn-SIGDA-75}. The variable gadgets form pairs of horizontal paths across a grid, and the clause gadgets form triples of vertical paths across the grid. The intended solution curve chooses one out of the two or three paths from each gadget, zigzagging horizontally in sequence through all the variable gadgets and then vertically in sequence through all the clause gadgets. See \cref{fig:redux-plan}, which draws the grid using horizontal and vertical lines in the plane. However, because we are placing points on a unit sphere, we will place this overall layout on a roughly square patch of the sphere, with these horizontal and vertical lines replaced by great-circle arcs. These arcs will not meet at exact right angles but this will not cause problems for our reduction.  The choice of which of the two paths to use for each variable gadget corresponds to the truth value assigned to that variable, and the choice of which of the three paths to use for each clause gadget corresponds to a selection of a variable whose truth assignment causes that clause gadget to become true.

We will choose the sample spacing parameter $\varepsilon$ to be sufficiently smaller than the spacing between these great circle arcs, ensuring that a polygonal curve that does not follow these arcs will have a too-sharp angle. To be more precise about what this means, suppose that we start with a 3-satisfiability instance with $n$ variables and clauses; then the overall form of the reduction from \cref{fig:redux-plan} can be embedded on a unit sphere, with arcs separated by distance $\Omega(1/n)$. Thus, we may choose $\varepsilon$ to be an appropriately small multiple of $1/n$. The layout of \cref{fig:redux-plan} involves tight curves between consecutive horizontal or vertical great-circle arcs, with radius $\Theta(1/n)$, and along these curves the points should be sampled more densely, with spacing $O(1/n^2)$, to ensure that a polygonal curve following these curves does not have too-sharp angles. Near the points of our layout where several curves meet smoothly, It may be possible for a polygonal curve with angles above the threshold $\theta$ to jump back and forth between these curves, but this is not problematic, because it still must end on a single curve through the main part of the gadget. We will also allow more dense sampling between the gadgets where two arcs cross, in order to  align the spacing of the samples within each crossing gadget.

\begin{figure}[t]
\centering\includegraphics[width=0.5\textwidth]{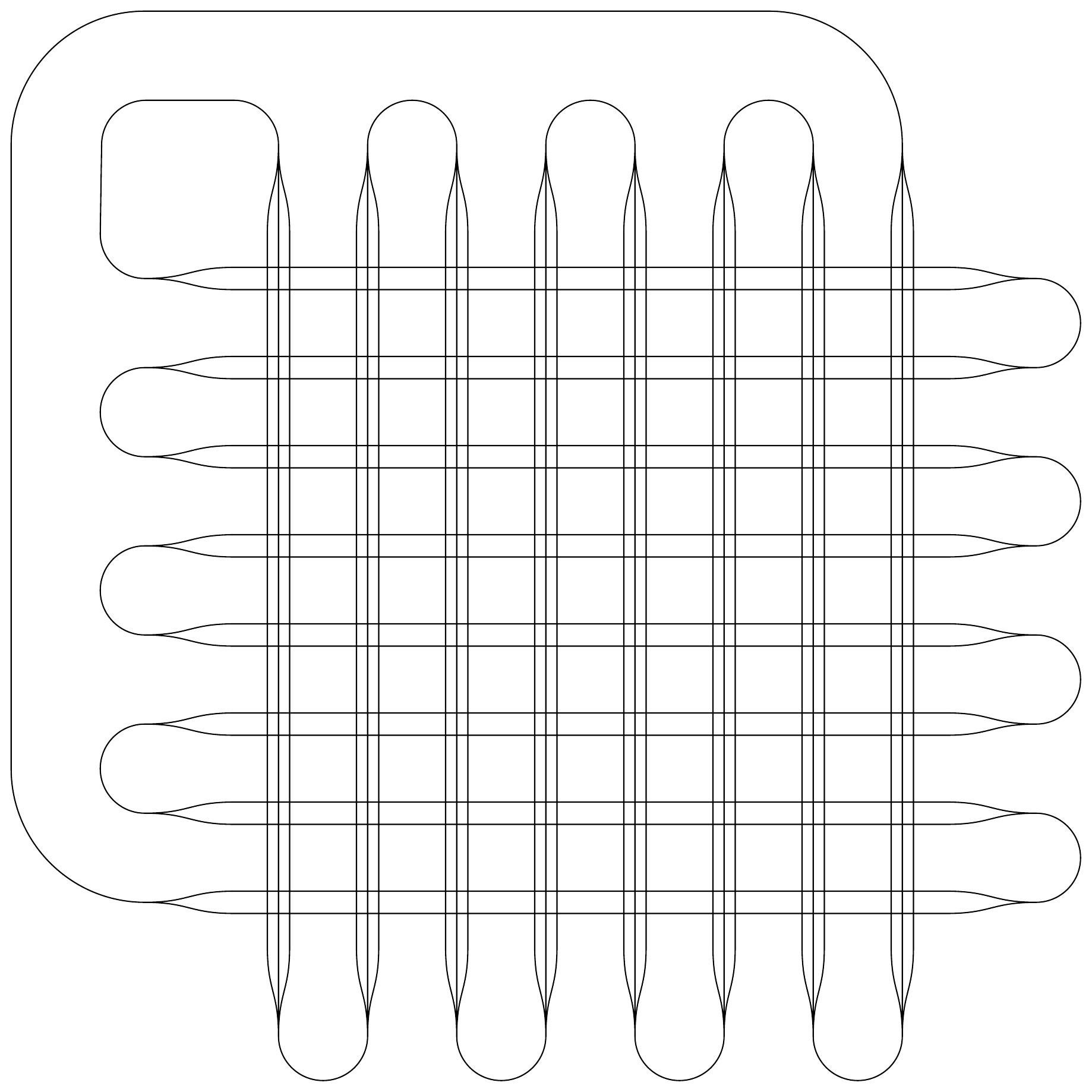}
\caption{Overall form of reduction: two horizontal paths across a grid for each variable gadget, and three vertical paths across the grid for each clause gadget, connected together to allow closed smooth curves that zigzag horizontally in sequence through each variable gadget and then zigzag vertically in sequence through each clause gadget}
\label{fig:redux-plan}
\end{figure}

It remains to describe what happens at the points where a variable gadget and a clause gadget cross each other. These are the only curve crossings in the entire construction. We desire two distinct behaviors for different types of crossing:
\begin{itemize}
\item At a crossing between the curve of a clause gadget for a clause $C$ that selects, a variable $V$ to be the variable satisfying $C$, and the curve of the variable gadget for $V$ that assigns $V$ a value incompatible with satisfying clause $C$, we cannot allow both curves to be used. We will replace this crossing with a \emph{blocking gadget} that allows one or the other curve to be used in a high-angle polygonal curve, but not both.
\item At all other crossings, it must be possible to use either curve without interference from the other. We will replace this crossing with a \emph{crossing gadget} that allows either or both curves to be used.
\end{itemize}

\begin{figure}[t]
\centering\includegraphics[width=0.5\textwidth]{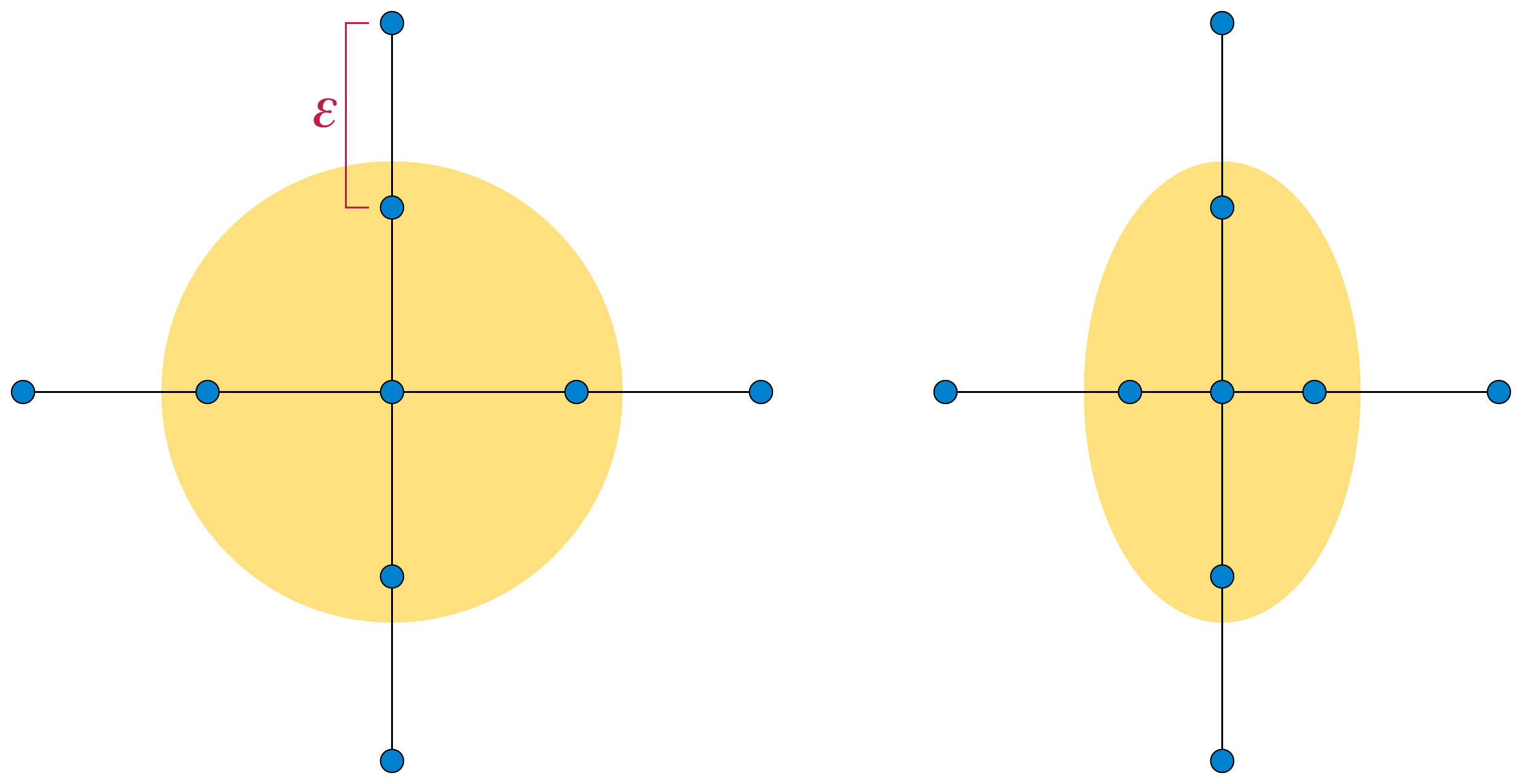}
\caption{Blocking gadget (left) and crossing gadget (right)}
\label{fig:crossings}
\end{figure}

These two gadgets are shown in \cref{fig:crossings}. They are depicted as embedded in a flat plane, but in fact will be embedded on the surface of the unit sphere, so the distances along each great-circle arc labeled as $\varepsilon$ will be as shown, but the other distances are not to scale. A polygonal curve that follows a great-circle arc through one of these crossing gadgets must use sample points along the arc spaced at distance at most $\varepsilon$ from each other; skipping any of these points would lead to an angle sharper than the threshold $\theta$. In the blocking gadget, at left, polygonal curves on either great-circle arc must use all points in the gadget, including the center point, preventing a non-self-intersecting curve from using both arcs. In the crossing gadget, at right, the horizontal curve can skip the center point, passing below the vertical path through the center point; therefore both arcs may be used by a single non-self-intersecting polygonal curve. Thus, with these gadgets placed as described above, the initial 3-satisfiability instance will have a satisfying assignment if and only if it is possible to choose which curve to follow through each variable and clause gadget in such a way that no two chosen curves pass through the same blocking gadget, producing a non-self-intersecting closed polygonal curve of minimum angle $\ge\theta$. On the other hand, if the satisfiability instance is not solvable, then any non-self-intersecting closed polygonal curve  must either follow both curves through a blocking gadget or jump from one great-circle arc to another rather than following the arcs of the layout, in either case producing an angle $<\theta$.

Recall that, in our construction, the curves connecting pairs of great-circle arcs in our construction have radius $O(1/n)$ and therefore need samples spaced at distance $O(1/n^2)$ in order to achieve angles greater than our threshold $\theta=\pi-\Theta(1/n)$.  We impose an additional requirement that the spacing of all pairs of consecutive sample points in our construction is $\Omega(1/n^2)$, preventing two sample points from being at very close distances along the curve. With this additional assumption, any sample point may be perturbed by a distance of $O(1/n^3)$ while still only changing any angle formed by three sample points by $O(1/n)$, small enough to preserve the gap between the angles in a valid solution and an invalid solution. Thus, we may construct our final instance by choosing a suitably large multiple $d$ of $n^3$, and rounding the Euler angles of each sample point to the nearest integer multiple of $1/d$. In this way, the rational numbers used to specify the Euler angles for these points need only $O(\log n)$ bits of precision each. With an overall layout consisting of $\Theta(n)$ great-circle arcs of constant length sampled by points spaced at a distance $\Theta(1/n)$ apart from each other, $\Theta(n^2)$ crossing and blocking gadgets with $O(1)$ sample points, and $\Theta(n)$ connecting curves of length $O(1/n)$ sampled by points at distance $\Theta(1/n^2)$ apart from each other, the total number of points constructed in this reduction is $\Theta(n^2)$, a polynomial.

This completes the proof. We observe that, if the crossing gadgets are arranged so that the variable-gadget curves all pass under the clause-gadget curves, the polygon corresponding to a satisfying assignment (if it exists) will be unknotted, so this also proves that finding the maxmin-angle unknotted polygon in 3d is $\mathsf{NP}$-hard.

\section{Cycles in graphs}
\label{sec:graphs}

In this section we explore bottleneck cycles in undirected graphs, directed graphs, mixed graphs, and polar graphs. In each case, we seek a cycle (of a type appropriate to the graph) whose maximum edge weight is minimized. These different types of graphs lead to problems of increasing difficulty, which can all nevertheless be solved in linear time. We will use our algorithms for cycles in directed and polar graphs as subroutines for some of our geometric problems.

In all of these problems we may define the quality of an edge to be its weight, if it belongs to a cycle in the remaining subgraph, or $-\infty$, if it does not. (However some of our algorithms can be interpreted as using a relaxation of this definition that sets some but not all weights of non-cycle edges to $-\infty$.) This is obviously monotone: the quality of an edge remains fixed until dropping to $-\infty$ once all cycles containing the edge have been broken. The edges of a bottleneck cycle will keep their weights until all lower-priority edges have been removed, so that the quality of the cycle as a subgraph is its bottleneck weight. Any subgraph of finite quality includes a cycle through every edge, and in particular it includes a cycle through its bottleneck edge, so although an optimal bottleneck subset may not itself be a cycle it will always include a bottleneck cycle.

\subsection{Undirected graphs}

Finding bottleneck cycles in an undirected graph is closely related to the complexity of maximum spanning trees. Any cycle must include at least one edge not in the maximum spanning tree, its bottleneck is the lightest such edge,
and each edge that is not in the maximum spanning tree is the lightest edge of at least one cycle (the cycle it induces in the tree). Thus, the bottleneck cycle is obtained from the heaviest edge not in the maximum spanning tree, and the cycle that it induces in the tree. It can be found in linear time once the maximum spanning tree is known.

However, this problem has a much simpler solution, suitable for an undergraduate exercise and similar  to a standard algorithm for the bottleneck spanning tree problem~\cite{Cam-IPL-78}, which we present briefly here for completeness. It is not a decremental greedy algorithm. Instead, find in linear time the median edge weight, and test whether there exists a cycle in the edges of its weight and higher. If so, delete the lighter edges and recurse. If not, contract the heavier edges and recurse. Once the recursion bottoms out with the bottleneck edge, find a path between its endpoints using heavier edges of the original graph to complete it into a cycle.

\subsection{Directed graphs}
\label{sec:directed}

Our algorithm for directed bottleneck cycles is based on a method of Gabow and Tarjan for finding bottleneck shortest path trees in directed graphs~\cite{GabTar-Algs-88}. This method repeatedly clusters the edge weights into blocks, rounds the weights to integer block indexes, and applies a fast integer algorithm to the rounded weights. Gabow and Tarjan applied their clustering method to a modified form of Dijkstra's algorithm for bottleneck paths in graphs with integer weights. We instead apply it to a different integer-weight algorithm, a linear-time decremental greedy algorithm for bottleneck cycles in directed graphs with integer (or presorted) weights.

\begin{lemma}
A bottleneck cycle in an edge-weighted directed graph with $m$ edges with sorted edge weights can be found in linear time by a decremental greedy algorithm.
\end{lemma}

\begin{proof}
In each iteration, remove a vertex with in-degree or out-degree zero, if one exists. If there is no such vertex then remove an edge of minimum weight, remembering its weight as the best seen so far. (This is the decremental greedy algorithm for a quality function that is either the weight of an edge or $-\infty$ for edges incident to degree-zero vertices.) The subgraph just before the last minimum-weight removal must be strongly connected, as otherwise it would have at least two strongly connected components that are sources or sinks (with no incoming or outgoing edges respectively), each such component would contain a cycle, and at least one of these two or more disjoint cycles would entail a later minimum-weight removal. Therefore, the last remembered edge is part of a cycle of heavier edges, a bottleneck cycle that can be obtained by finding a path in the remaining subgraph between its two endpoints.

The steps in which a minimum-weight edge is found can be performed in linear total time by scanning the sorted list of edges, starting the scan from the previously remembered minimum-weight edge and skipping past edges that have already been removed. The remaining steps can be performed in constant time each by maintaining a ready-list of degree-zero vertices, keeping track of the degree of each vertex after each removal, and moving a vertex into the ready-list when its degree becomes zero. Therefore, for a graph whose edge weights are presorted or are small-enough integers that they can be sorted in linear time, the total time is linear.
\end{proof}

\begin{theorem}
\label{thm:directed-cycle}
A bottleneck cycle in an edge-weighted directed graph with $m$ edges can be found in time $O(m\log^* m)$, accessing the edge weights only by comparisons.
\end{theorem}

\begin{proof}
We perform a sequence of runs of the linear-time decremental greedy algorithm, in each run narrowing the range of edge weights known to contain the bottleneck value. We control the number of edges remaining in this bottleneck weight range by a sequence of parameters $\alpha_i$, with $\alpha_1=1$ and $\alpha_{i+1}=2^{\alpha_i}$. In the $i$th run, we will already have found a block $B_i$ of at most $m/\alpha_i$ edges that might contain the bottleneck value, consecutive in the sorted order of edge weights, with the edges outside this block having weights that are known to be above or below the bottleneck value. As a base case, this is trivially true for $\alpha_1$, as the single block contains all edges. We use recursive median-finding to refine block $B_i$ into a sorted sequence of sub-blocks of at most $m/\alpha_{i+1}$ edges each. Because of how $\alpha_{i+1}$ was chosen, the recursive median-finding step takes total time $O\bigl((m/\alpha_i)\log\alpha_{i+1}\bigr)=O(m)$. We number the sub-blocks in sorted order, and number the edges by the sub-block they belong to, with edges above or below the current block assigned equal numbers above or below the sub-block numbers, respectively. Then we sort the edges by these small integers and apply the linear-time decremental greedy algorithm for sorted sequences, obtaining as a result the identity of the sub-block containing the bottleneck value.

After $\log^*m$ iterations of this $O(m)$-time refinement process, we will have $\alpha_i\ge m$, at which point we have identified a single edge known to be the bottleneck value. At this point we can delete all lower-weight edges and find a cycle in the remaining graph.
\end{proof}

\subsection{Mixed graphs}

A \emph{mixed graph} is allowed to contain both directed and undirected edges. A cycle, in such a graph, can use the undirected edges in either direction, but all of its directed edges must be consistently oriented. A \emph{walk} in a mixed graph allows repeated edges and vertices, with the same restriction on consistent orientation of directed edges, but allowing a single undirected edge to be traversed inconsistently, in both directions. We define a \emph{u-turn} in a walk to be a pair of steps that use the same undirected edge consecutively in opposite directions (these have also been called \emph{backtracks}). We say that the \emph{u-turn} is at a vertex $v$ if $v$ is the endpoint of the undirected edge that occurs between its two consecutive uses. Replacing every undirected edge of a mixed graph by two directed edges preserves all walks (by a one-to-one correspondence), but we will use a modified form of this replacement that prevents u-turns. We define walks as having an orientation, even when they use only undirected edges, so that for instance every undirected cycle of the undirected subgraph of a mixed graph corresponds to two walks.

\begin{lemma}
\label{lem:mixed-cycle}
A mixed graph has a cycle if and only if it has a closed u-turn-free walk.
\end{lemma}

\begin{proof}
If it has a cycle then that cycle is a walk, so need only show that a mixed graph $G$ that has a u-turn-free walk $W$ has a cycle. If the undirected subgraph of $G$ has a cycle, the result is immediate, so we may assume without loss of generality that this subgraph is acyclic, and that $W$ contains at least one directed edge. Let $G'$ be the directed multigraph obtained from $G$ by deleting all edges that do not participate in $W$ and by contracting all undirected edges in $W$. Then every vertex of $G'$ has both incoming and outgoing edges, so $G'$ cannot be a DAG and it contains at least one cycle $C'$. From $C'$, obtain a cycle $C$ in $G$ by uncontracting the undirected components of $G$ and replacing each vertex in $C'$ that comes from a contracted undirected component by a path through that component.
\end{proof}

\begin{figure*}[t]
\centering\includegraphics[scale=0.25]{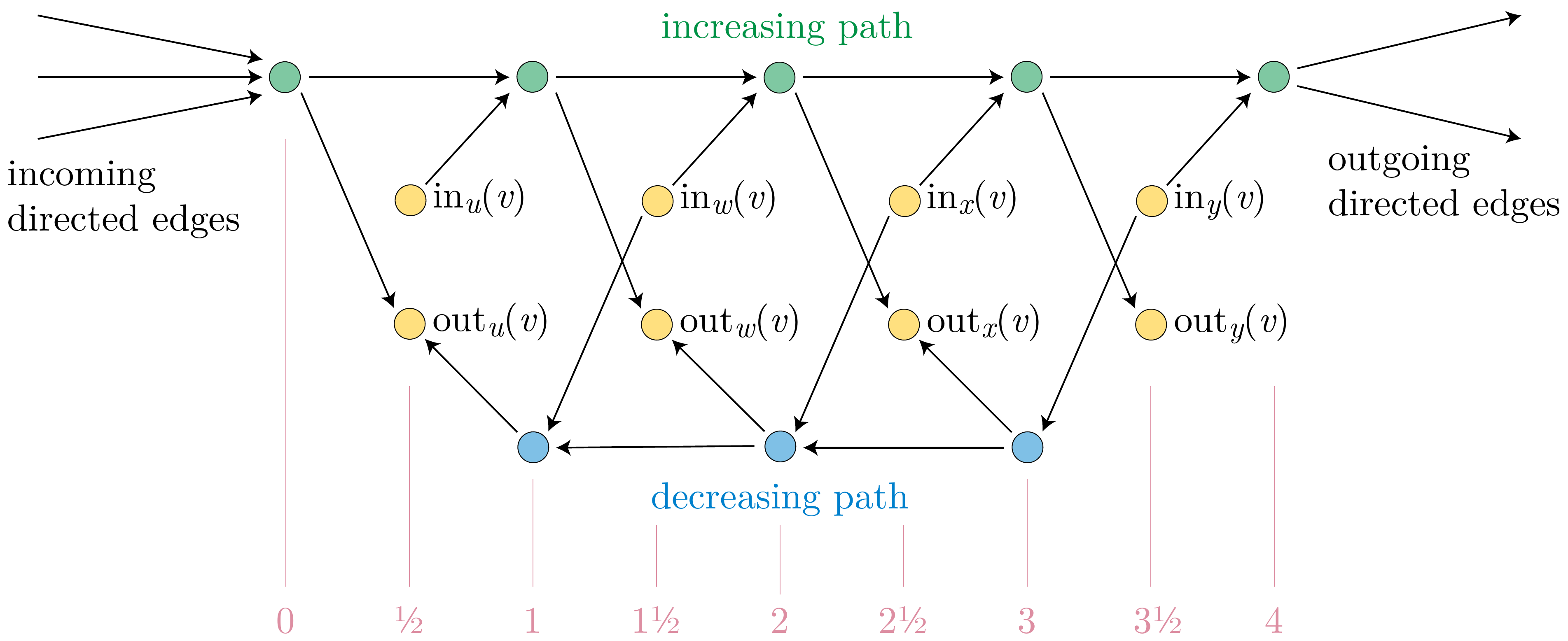}
\caption{Expansion of a mixed graph vertex $v$, with incident undirected edges $uv$, $wv$, $xv$, and $yv$, into a subgraph of a directed graph (\cref{lem:no-u-turn})}
\label{fig:mixed-expansion}
\end{figure*}

\begin{lemma}
\label{lem:no-u-turn}
It is possible to transform a mixed graph $G$ with $m$ edges into a directed graph with $O(m)$ edges, in such a way that the closed walks of the resulting directed graph correspond one-for-one with the closed walks of the mixed graph that do not have u-turns.
\end{lemma}

\begin{proof}
We show how to expand each vertex $v$ into a gadget that prevents $u$-turns at $v$. Let the undirected degree of $v$ be $d_u(v)$. Then we replace $v$ by the following vertices:
\begin{itemize}
\item Two vertices for each undirected edge $uv$, an \emph{incoming terminal} $\mathrm{in}_u(v)$ (to which the incoming directed edge corresponding to $uv$ attaches) and an \emph{outgoing terminal} $\mathrm{out}_u(v)$ (to which the outgoing directed edge attaches).
\item A directed path of $d+1$ vertices, the \emph{increasing path}. We will attach all incoming directed edges of $v$ in $G$ to the starting vertex of the increasing path, and all outgoing directed edges of $v$ in $G$ from the ending vertex of the increasing path.
\item Another directed path of $d-1$ vertices, the \emph{decreasing path}.
\end{itemize}
In order to preserve closed walks, we need to connect these vertices by additional directed edges in such a way that there are no closed walks within $v$, and so that every incoming edge to $v$ has a unique path to every outgoing edge from $v$, except that there should be no path from $\mathrm{in}_u(v)$ to $\mathrm{out}_u(v)$. To do so, we number the vertices of the increasing path from $0$ to $d$, we number the vertices of the decreasing path in descending order from $d-1$ to~$1$, and we assign distinct half-integers to the undirected edges $uv$, arbitrarily, from $\tfrac12$ to $d-\tfrac12$. We connect each vertex $\mathrm{in}_u(v)$, with number $x$, by edges to vertex $x+\tfrac12$ of the increasing path and to vertex $x-\tfrac12$ of the decreasing path, so that it can reach everything with a larger number on the increasing path and everything with a smaller number on the decreasing path. Symmetrically, we connect each vertex $\mathrm{out}_u(v)$, with number $x$, by edges from vertex $x-\tfrac12$ of the increasing path and from vertex $x-\tfrac12$ of the decreasing path, so that it can be reached by everything with a smaller number on the increasing path and by everything with a larger number on the decreasing path. In this way, connections in~$G$ through~$v$ that use a directed edge (either into or out of $v$) can be made through the increasing path. Connections through~$v$ that come into $v$ on one undirected edge $uv$ and then leave through another undirected edge $vw$ can be made on the increasing path, if $uv$ is numbered less than $vw$, and on the decreasing path, otherwise. The construction is illustrated in \cref{fig:mixed-expansion}.
\end{proof}

\begin{theorem}
A bottleneck cycle in an edge-weighted mixed graph with $m$ edges can be found in time $O(m\log^* m)$, accessing the edge weights only by comparisons.
\end{theorem}

\begin{proof}
Transform the graph into a directed graph, preserving $u$-turn-free walks but eliminating u-turns, by \cref{lem:no-u-turn}. Then apply \cref{thm:directed-cycle} to the resulting directed graph.
\end{proof}

\subsection{Polar graphs}
\label{sec:polar}

We use the concept of a \emph{polar graph} of Zelinka~\cite{Zel-AM-74}, later called a switch graph~\cite{Coo-NCiCA-03}. This is a graph in which each vertex has two \emph{poles}, with each incident edge assigned to one of the two poles. A \emph{regular path} or \emph{regular cycle} in such a graph must avoid repeated vertices, and at each intermediate vertex of the path or cycle it must enter and exit through opposite poles. Polar graphs can be expanded into \emph{skew-symmetric graphs}, directed graphs with two copies of each polar graph vertex, one for each direction of travel through it, in which regular paths can be found in linear time using a generalization of alternating path search~\cite{GolKar-Comb-96}. A regular cycle can then be found, if one exists, by using this path-finding algorithm to try to extend each edge into a cycle. Alternatively, in a polar graph of maximum degree three, a greedy algorithm that repeatedly removes bridge edges and vertices with edges incident to only one pole (neither of which can participate in a regular cycle) reduces the input to a subgraph in which each nontrivial component contains a regular cycle~\cite{Coo-NCiCA-03}. We will use both of these ideas as part of our decremental greedy algorithm for bottleneck regular cycles.

\begin{figure}[t]
\centering\includegraphics[width=\columnwidth]{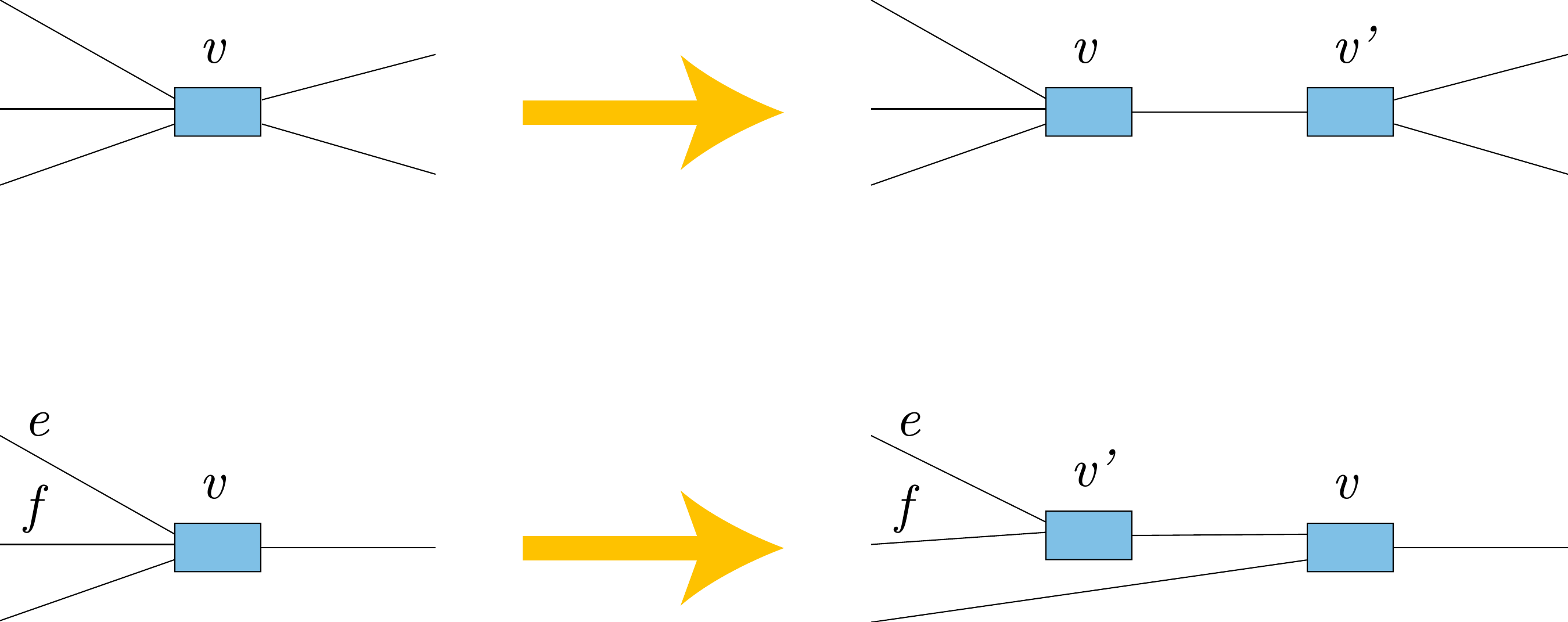}
\caption{Expanding a polar graph into a degree-three polar graph (\cref{lem:polar-tree}). The two poles of each vertex are depicted as the opposite sides of a rectangle.}
\label{fig:polar-tree}
\end{figure}

\begin{lemma}
\label{lem:polar-tree}
In a polar graph $G$, it is possible to replace any vertex $v$ of total degree $d>3$ by a tree of $d-2$ vertices and $d-1$ additional edges, in such a way that a set of edges forms a regular cycle in $G$ if and only of the same set, together with a unique path in the added edges, forms a regular cycle in the expanded graph.
\end{lemma}

\begin{proof}
We show that we can replace $v$ by two vertices connected by a single edge. The full result follows by an induction on $d$ that performs this vertex split operation and then applies the induction hypothesis to both split vertices.

To split $v$ into two vertices connected by an edge, If $v$ has at least two edges at both of its poles, split its poles into two separate vertices,  and add an edge between the opposite poles of these two new vertices (\cref{fig:polar-tree}, top). Every regular cycle through $v$ in $G$ corresponds uniquely to a cycle through both of these two vertices and the added edge. Other regular cycles are unchanged.

In the remaining case, $v$ has only one edge attached to one pole, and three or more edges attached to the other pole.
Let $e$ and $f$ be two edges attached to the same pole of $v$. Add a new vertex $v'$, connected to~$v$ at the same pole as $e$ and $f$, and then reattach $e$ and~$f$ from $v$ to the opposite pole of $v'$ (\cref{fig:polar-tree}, bottom). This reduces the degree of $v$ by one and adds one new degree-three vertex. The regular cycles through $e$ or $f$ in $G$ correspond uniquely to regular cycles through $v'$ in the expanded graph, and other regular cycles are unchanged.
\end{proof}

Call the replacement of \cref{lem:polar-tree} a \emph{polar tree}. If we assign high weight to all edges in a polar tree, it will not change the bottleneck edge or quality of any regular cycle.

\begin{theorem}
In a polar graph with $m$ edges, a bottleneck regular cycle can be found in time $O\bigl(m(\log m\log\log m)^2\bigr)$.
\end{theorem}

\begin{proof}
Expand the graph into a polar graph of degree three by \cref{lem:polar-tree}. Then, use a greedy algorithm that repeatedly removes vertices and edges, with the following priority:
\begin{itemize}
\item If there is a vertex with edges incident only to a single pole, remove it and all its incident edges.
\item If there is a bridge edge (an edge whose removal would disconnect the underlying undirected graph of the current polar graph), remove it.
\item If neither of the two previous cases applies, remove the lightest remaining edge.
\end{itemize}

The bridges can be found by applying applying dynamic graph algorithms for bridge-finding in the word RAM model~\cite{HolRotTho-SODA-18}, in time $O\bigl(m(\log m\log\log m)^2\bigr)$ per edge removal. When the lightest remaining edge is needed, it can be found by using a priority queue of remaining edges, in time $O(\log m)$ per edge removal.
The bottleneck edge is the the last one that edge was removed as a lightest remaining edge.
Once it is found, a regular cycle using it can be found by returning the graph to the state just before its removal, and then applying the linear-time algorithm for finding a regular path between the endpoints of the bottleneck edge.
\end{proof}

\putbib
\end{bibunit}

\fi

\end{document}